\documentclass[
	sigconf,
	nonacm=true,
	natbib=false
]{acmart}

\usepackage[T1]{fontenc}

\usepackage[utf8]{inputenc}

\usepackage{amsmath}

\usepackage{amsthm}

\usepackage{wasysym}

\usepackage{mathpartir}

\usepackage{hyperref}
\hypersetup{
	pdfauthor={\authors},
	pdftitle={\@title},
	pdfsubject={\@concepts},
	pdfkeywords={\@keywords},
	pdfcreator={LaTeX with acmart
			\csname ver@acmart.cls\endcsname\space
			and hyperref1835\csname ver@hyperref.sty\endcsname}}

\usepackage{csquotes}

\usepackage[
	backend=bibtex,
	sorting=none,
	style=ieee,
	url=false,
	doi=false,
	isbn=false
]{biblatex}
\usepackage{siunitx}
\usepackage{graphicx}
\graphicspath{{resources/img/}}

\usepackage{paralist}

\usepackage{caption}

\usepackage{listings}
\usepackage{algorithmicx}

\usepackage{algpseudocode}

\usepackage{booktabs}

\usepackage{multirow}

\usepackage{pifont}

\usepackage{subfig}

\usepackage[shortcuts]{extdash}

\usepackage{microtype}

\newtheorem{definition}{Definition}
\newtheorem{proposition}{Proposition}

\newcommand{\ba}{\begin{array}}
\newcommand{\ea}{\end{array}}
\newcommand{\bda}{\[\ba}
\newcommand{\eda}{\ea\]}

\newcommand{\hs}{\hspace{7.5pt}}                 %
\newcommand{\str}[1]{\texttt{"#1"}}              %
\newcommand{\attribute}[2]{(\str{#1}, \str{#2})} %

\newcommand{\equality}[2]{#1 \equiv #2}           %
\newcommand{\inequality}[2]{#1 \not\equiv #2}     %
\newcommand{\maintainedBy}[2]{#1 \approx #2}      %
\newcommand{\maintainedByI}[2]{#1 \approx_1 #2}   %
\newcommand{\maintainedByII}[2]{#1 \approx_2 #2}  %
\newcommand{\maintainedByIII}[2]{#1 \approx_3 #2} %

\newcommand{\filter}{\mathcal F}
\newcommand{\extract}{\mathcal E}
\newcommand{\projKV}[2]{#1 \Downarrow_{KV} #2}
\newcommand{\projK}[2]{#1 \Downarrow_{K} #2}

\newcommand{\rlabel}[1]{\mbox{(#1)}}   %

\DeclareMathOperator{\simjw}{sim_{JW}} %
\DeclareMathOperator{\eqv}{eq}         %
\DeclareMathOperator{\idente}{ident_{\approx, \extract}} %

\newcommand{\myirule}[2]{{\renewcommand{\arraystretch}{1.2}\ba{c} #1
                      \\ \hline #2 \ea}}

\DeclareSIUnit\inch{''}
\DeclareSIUnit\pixel{px}

\algnewcommand\Break{\State \textbf{break}}
\algnewcommand\Continue{\State \textbf{continue}}

\renewcommand*{\bibfont}{\footnotesize}

\begin{document}

\title{Visual Testing of GUIs by Abstraction}

\author{Daniel Kraus}
\email{daniel.kraus@retest.de}
\affiliation{%
	\institution{ReTest GmbH}
	\streetaddress{Haid-und-Neu-Straße 7}
	\postcode{76131}
	\city{Karlsruhe}
	\country{Germany}}

\author{Jeremias Rößler}
\email{jeremias.roessler@retest.de}
\affiliation{%
	\institution{ReTest GmbH}
	\streetaddress{Haid-und-Neu-Straße 7}
	\postcode{76131}
	\city{Karlsruhe}
	\country{Germany}}

\author{Martin Sulzmann}
\email{martin.sulzmann@hs-karlsruhe.de}
\affiliation{%
	\institution{Karlsruhe University of Applied Sciences}
	\department{Faculty of Computer Science and Business Information Systems}
	\streetaddress{Moltkestraße 30}
	\postcode{76133}
	\city{Karlsruhe}
	\country{Germany}}



\begin{abstract}
Ensuring the correct visual appearance of graphical user interfaces~(GUIs) is important
because visual bugs can cause substantial losses for businesses.
An application might behave functionally correct in an automated test,
but visual bugs can make the GUI effectively unusable for the user.
Most of today's approaches for visual testing are pixel-based and tend to have flaws
that are characteristic for image differencing.
For instance, minor and unimportant visual changes often cause false positives,
which confuse the user with unnecessary error reports.
Our idea is to introduce an abstract GUI state~(AGS),
where we define structural relations to identify relevant GUI changes and
ignore those that are unimportant from the user's point of view.
In addition, we explore several strategies to address
the GUI element identification problem in terms of AGS.
This allows us to provide rich diagnostic information
that help the user to better interpret changes.
Based on the principles of golden master testing,
we can support a fully-automated approach to visual testing by using the AGS.
We have implemented our approach to visually test web pages and
our experiments show that we are able to reliably detect GUI changes.
\end{abstract}

\maketitle

\section{Introduction}
\label{sec:intro}

\emph{Graphical user interfaces~(GUIs)} are ubiquitous.
Websites, for example, act as GUIs for the services behind.
Millions of users access the Internet through their browser every day,
interacting with a variety of different websites such as Amazon, Google, YouTube and many more.
As pointed out by Alameer et al.~\cite{alameer16},
users often base their impressions of trustworthiness
as well as quality---and ultimately the decision
to purchase a product---on the visual appearance~\cite{egger00, everard05, fogg01}.
Therefore, it is crucial to ensure the \emph{visual correctness} of such web-based GUIs.
The problem is that the system under test~(SUT) might behave functionally correct in an automated test,
but unnoticed visual bugs can make the GUI unusable from a user's perspective~\cite{li10}.
Visual bugs are not negligible as they can cause substantial losses for businesses~\cite{battat19}.
Consider, e.g., a paid ad not being displayed properly.
In the case of a pay-per-click platform such as Facebook,
both the advertising provider and the advertiser are affected.

Web-based GUIs are particularly challenging in the context of visual testing because they are being accessed in many different ways.
Users may visit a website from a desktop computer, a notebook, a tablet, or a smartphone.
In addition, different browsers---or different versions of the same browser---may render a web page differently.
This results in a multidimensional matrix that includes a wide range of devices, operating systems, browsers, screen sizes and resolutions.
Each entry in this matrix represents a possible usage scenario, potentially containing visual bugs, that the GUI developer must consider~\cite{althomali19}.

In the area of manual testing, some tools~\cite{luo18, quirktools19, testsize19} allow to visually inspect the GUI using a range of common screen sizes and resolutions.
While this is helpful for sanity checks or exploratory testing, manually checking for visual correctness is time consuming, inconsistent and prone to human errors~\cite{althomali19}.

In terms of automated testing, various approaches~\cite{galenframework19, halle16, panchekha18, zaiats19} are based on formal specifications,
which define desired GUI properties.
But writing and maintaining specifications can be tedious and requires manual effort.
Furthermore, a specification can only protect against expected changes~\cite{slatkin13};
namely those changes that are covered by the specification.
All other GUI changes are implicitly allowed as they do not lead to test failures.

Because of this, many visual testing approaches rely on \emph{golden master testing} (or characterization testing),
a means to characterize the behavior of the test object to protect it against unintended changes~\cite{feathers04}.
In order to do so, the results of a previous, typically stable version (the golden master) serve as the test oracle.
That is, a golden master test passes if the corresponding behavior of the test object remains unchanged---regardless of its correctness.
Most implementations here,
both academic~\cite{mahajan16, mahajan15, roychoudhary13, saar16} and industrial~\cite{applitools19a, slatkin16, percy19},
use image comparison.

\begin{figure*}
	\centering
	\includegraphics{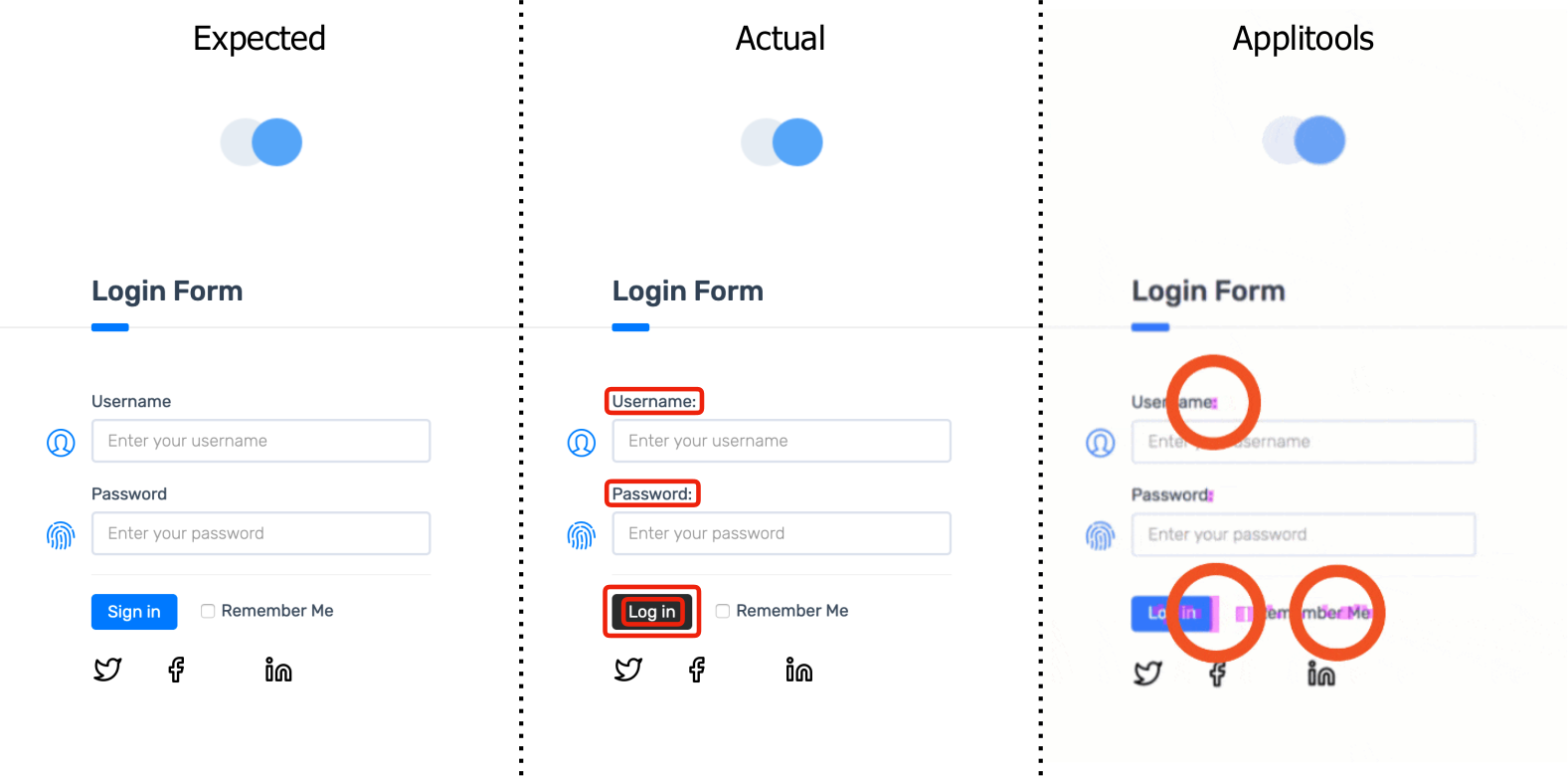}
	\subfloat[\label{fig:login-expected} Expected login screen.]{\hspace{.23\linewidth}}
	\subfloat[\label{fig:login-actual} Actual differences.]{\hspace{.23\linewidth}}
	\subfloat[\label{fig:login-applitools} Applitools' differences.]{\hspace{.23\linewidth}}
	\caption{Expected compared to actual login screen and differences highlighted by Applitools~\cite{applitools19b}.}
	\label{fig:login-all}
\end{figure*}

The problem with pixel-based tools is that they often produce false positives when minor, unimportant visual differences occur~\cite{alegroth18}.
For instance, a new browser version may cause screenshots to be off by only a few pixels, which often results in test failures.
Moreover, the majority of these tools have problems grouping together the same or similar changes.
If, e.g., the header of a website changes and this header is being used across multiple tests,
all of these tests usually have to be maintained one by one, although it is the very same underlying change.
Also, the vast majority of pixel-based tools only report that something in a certain area has changed,
not \emph{how} it has changed (e.g. text or layout change).
This means additional effort for the users to interpret the given GUI changes.

Instead of using a pixel-based representation, our idea is to find a suitable abstraction for a GUI.
Inspired by existing approaches that incorporate some kind of tree-based structure for GUI testing~\cite{moran18, grechanik18, walsh17, xie09},
we present an \emph{abstract GUI state~(AGS)} that is:
\begin{enumerate}
	\item Sufficiently expressive to capture the essence of a GUI including visible and non-visible properties.
	\item Computationally tractable to compare two (expected vs. actual) GUI states and to identify any changes.
	\item Customizable to let the user decide what changes are important and which are not.
	\item Platform-independent to achieve a high degree of reusability.
\end{enumerate}
On top of the AGS, we support a highly-automated method for visual testing by abstraction.
We implement the AGS-centric components and leverage the principles of golden master testing
to create an extensible test framework called \emph{recheck}~\cite{retest19a}.
To demonstrate recheck's capabilities, we provide \emph{recheck-web}~\cite{retest19b};
an adapter for web-based GUIs to the AGS.

For the evaluation of our approach, we created a visual testing benchmark~\cite{kraus19}
that contains 20 of the most popular websites,
which we use to simulate typical GUI changes.
In addition, we compare our implementation against Applitools~\cite{applitools19a},
an industrial tool for visual testing.
As part of our benchmark repository,
we also document the entire setup in detail and
publish all relevant data that was collected during our experiments.
recheck, recheck-web as well as the benchmark are fully open source.

In summary, our contributions are:
\begin{description}
	\item[C$_1$] We introduce a platform-independent representation of GUI states that allows us to define structural relations
			to detect GUI changes.
	\item[C$_2$] We specify the problem of GUI element identification and explore several strategies for it
			to provide the user with rich diagnostic information.
	\item[C$_3$] We implement the approach as a platform-independent test framework and provide an adapter for web-based GUIs,
			both open source and freely available.
	\item[C$_4$] We create an open-source benchmark for visual testing of web-based GUIs,
			which we use to evaluate our approach.
\end{description}
C$_1$ and C$_2$ are covered in Section~\ref{sec:ags},
C$_3$ and C$_4$ in Section~\ref{sec:eval}.
We discuss related work in Section~\ref{sec:related-work} and conclude our results in Section~\ref{sec:conclusion-future-work}.
The upcoming section starts with a general overview of our approach.

\section{Overview}
\label{sec:overview}

We follow earlier works and employ a tree-structured representation for GUIs,
to which we refer as \emph{abstract GUI state~(AGS)}.
Based on the principles of golden master testing,
we compare the \emph{actual} GUI state against an \emph{expected} GUI state.
Both states are represented as AGSes.
The comparison test is carried out by identifying deleted, created and maintained elements,
a.k.a. the \emph{GUI element identification problem}~\cite{mcmaster09}.
Maintained elements are not necessarily identical,
but functionally equivalent.
For this purpose, we explore several strategies on top of the AGS.
Using these abstractions,
we can identify how and why GUIs differ.

Compared to other approaches,
we are able to provide detailed diagnostic information to the user.
This helps to better interpret GUI changes and assists the user when reviewing them.
In the following, we give an overview of our approach.
We do this by first illustrating the current state of visual testing using Applitools,
a well-known industrial tool, as an example.
Then we outline several improvements of our approach when compared to pixel-based methods.

Applitools offers various APIs for major programming languages and test frameworks to create tests with visual checks.
During the execution of a test, every check creates a screenshot that is uploaded to a dedicated service,
where it is compared against a given baseline (i.e. expected vs. actual).
Computer vision~(CV) algorithms attempt to only report perceptible differences,
an advantage over conventional pixel-by-pixel approaches.
A test manager can be used to review the test results in detail.
Changes can be either approved, rejected, or ignored.

Figure~\ref{fig:login-all} shows three screenshots of a demo login screen used in one of Applitools' tutorials~\cite{applitools19b}.
The left screenshot \ref{fig:login-expected} represents the initial version of the login (expected).
The middle screenshot \ref{fig:login-actual} is an adapted version (actual).
(We ignore the right part for now.)
Suppose expected is the login screen used in production, whereas actual is a modified version that should be released.
But before the actual deployment, we want to make sure that only those parts of the GUI have changed that are intended to change.
As highlighted in red in the middle, there are four \emph{relevant} differences:
\begin{enumerate}
	\item The label \enquote{Username} is now followed by a colon (\enquote{:}).
	\item Also the label \enquote{Password} is now followed by a colon.
	\item The text of the \enquote{Sign in} button now is \enquote{Log in}.
	\item The color of this button changed from blue to black.
\end{enumerate}
Let's say all of these changes were intended.
The problem with traditional, pixel-based approaches like Applitools is that they often report false positives.
For instance, the \enquote{Remember Me} checkbox slightly moved to the left, caused by the shorter \enquote{Log in} button.
One could argue that this is something worth reporting, but most of the time this is just noise for the user.
Both the \enquote{Log in} button and the \enquote{Remember Me} checkbox are still aligned, the distance between remains the same, and there is no visual bug such as an overlap.
Consequently, we can assume this is an irrelevant difference.
Nevertheless, the change is reported by Applitools
as can be seen on the right in Figure~\ref{fig:login-applitools},
which is an actual screenshot from the tutorial~\cite{applitools19b}.

Also recent advances suffer from similar flaws.
For instance, the authors of GCat~\cite{moran18} are pointing out that it is difficult to properly tune the sensitivity of the image comparison.
If the threshold for reporting differences is too low, it might cause false positive.
If it is too high, false negatives may occur.

\begin{figure*}
	\centering
	\includegraphics[scale=0.2]{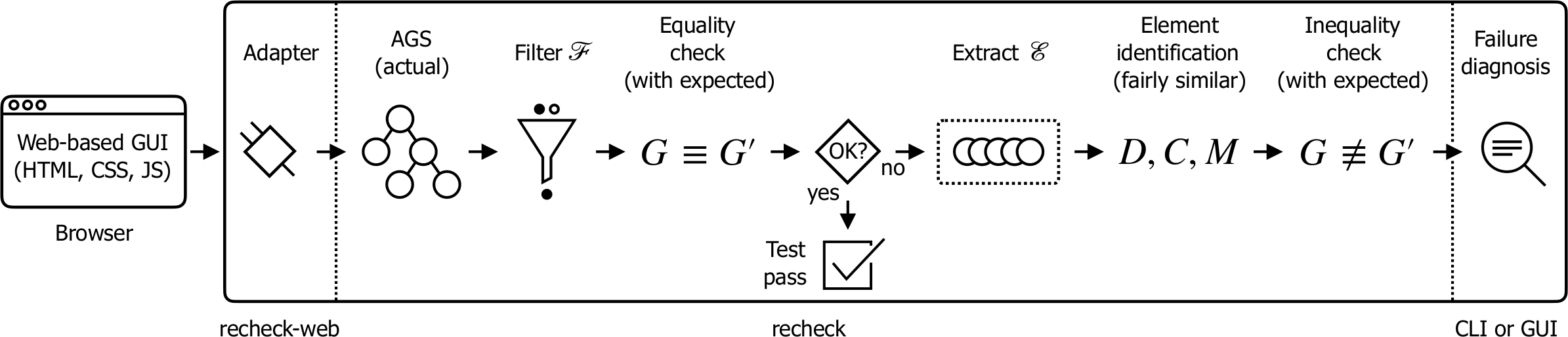}
	\caption{The AGS framework and its tools.}
	\label{fig:ags-framework}
\end{figure*}

Our idea is to introduce a rich but platform-independent representation for GUI states---the AGS.
It is structure-preserving, captures both visible and non-visible attributes and is obtained via a platform-specific adapter.
As an example, let's look at the HTML code of the former \enquote{Sign in} button:
\begin{verbatim}
<a id="login" class="btn btn-primary"
   href="/app.html">Sign in</a>
\end{verbatim}
The button is represented by an \texttt{<a>} element with attributes \texttt{id}, \texttt{class} and a \texttt{href} as well as the text \str{Sign in}.
The browser analyzes the given HTML and constructs the document object model~(DOM)---the in-memory representation of a web page.
Also CSS attributes and JavaScript code possibly influence the DOM and,
therefore, the button's appearance.
For simplification, we ignore most of it for our example.
Using our AGS syntax, this would result in the following data structure:
\bda{lll}
E & = & (\{ \attribute{id}{login}, \\
  &   & \hs \attribute{background-color}{\#047bf8}, \\
  &   & \hs \attribute{href}{/app.html}, \\
  &   & \hs \attribute{text}{Sign in}, \\
  &   & \hs \attribute{type}{a} \}, [])
\eda
The HTML \texttt{<a>} element has been mapped to an AGS element tuple $E = (As, Es)$,
consisting of a set of attributes $As = \{A_1, \dots, A_n\}$ and an empty list of child elements $Es = []$.
An attribute $A$ simply refers to a key-value pair $(K, V)$.
Again, we broke down this mapping to its essential parts,
leaving out other attributes derived from the underlying data like
the text color or the element position and size.
With these simple building blocks, we can capture the essence of a GUI.
In case of the new \enquote{Log in} button, differences are reflected in the element's attributes:
\bda{lll}
E' & = & (\{ \attribute{id}{login}, \\
  &   & \hs \attribute{background-color}{\#292b2c}, \\
  &   & \hs \attribute{onclick}{login()}, \\
  &   & \hs \attribute{text}{Log in}, \\
  &   & \hs \attribute{type}{button} \}, [])
\eda
Based on the visual appearance, one would only expect \texttt{text} and \texttt{background-color} to change.
But the AGS reveals that also the type has changed from \texttt{<a>} to \texttt{<button>},
which affects the implementation of the login:
instead of an \texttt{href} target, we now have an \texttt{onclick} event listener attached.

Information like this can be helpful for the user, but it depends on the context.
In fact, as we also track changes of an element's position and size,
one would first assume that we also report false positives like the shift of the \enquote{Remember Me} checkbox mentioned above.
But once the AGS has been constructed and before we attempt to compare expected and actual,
the user can filter specific elements and attributes to permanently ignore small changes like this.

We do this by defining a \emph{filter} function $\filter$ that,
when applied to a GUI state $G$, yields a possibly filtered GUI state $G'$.
The elements in $G'$ are not necessarily all of $G$ and may not contain all of their attributes.
While the mapping to the AGS is platform-specific, this step is domain-specific.
For example, we can filter out position changes by defining a certain threshold.
We can also filter when the GUI contains dynamic content such as animations.
We can even filter out the entire text of a GUI to test the effects of internationalization to the visual appearance.
Image processing is typically too fragile for this~\cite{alameer16},
and other approaches that rely on DOM analysis often raise false positives on different DOMs,
although they may lead to identical or sufficiently similar renderings~\cite{saar16}.
Thus, the user can utilize $\filter$ to precisely specify when changes are to be reported,
and to implement cross-browser and \=/platform visual testing.

The next step constitutes the actual comparison of the expected and actual GUI state.
We use an \emph{equality} relation $\equality{}{}$ to check if two AGSes are the same,
formulated as a declarative proof system.
We refer to the expected login screen as $G$ and to the actual login screen as $G'$.
We then attempt to verify $\equality{G}{G'}$ by recursively
comparing each element from $G$ to its corresponding element in $G'$.
In case of the login button,
this leads to the situation where $K = K' = \str{text}$,
but $V = \str{Sign in}$ and $V' = \str{Log in}$.

If equality fails like here,
we first attempt to identify GUI elements between expected an actual, i.e.,
to address the aforementioned GUI element identification problem.
For this, we employ a less strict \emph{fairly similar} relation $\maintainedBy{}{}$,
which identifies deleted and created elements as well as similar (maintained) pairs via customizable strategies.
For instance, the former \enquote{Sign in} button and the new \enquote{Log in} button
are considered a fairly similar element pair,
since the button has been maintained but is functionally equivalent.

Note that we not necessarily include all elements;
depending on the application context,
we might only want to consider leaf elements for example.
For this purpose, we use an \emph{extractor function} $\extract$
that selects the elements to be analyzed.

In a final step,
we provide diagnostic information by introducing an
\emph{inequality} relation $\inequality{}{}$ that explicitly captures the failure reason.
We apply inequality on maintained pairs and
report all found differences afterwards (including deleted and created elements).
In our example, this yields differences for the keys
\texttt{text}, \texttt{type}, \texttt{background-color}, \texttt{href} and \texttt{onclick}.

To summarize, the AGS essentially represents a mapping from a specific GUI technology to a generic GUI abstraction.
Based on the AGS, we have built a declarative formalism to specify equality, inequality and fairly similar in a platform-independent manner.
Given this foundation, we are able to provide a powerful framework for automated visual testing using pure abstraction.
Figure~\ref{fig:ags-framework} illustrates the essential components of our approach and
the individual steps of our golden master-based testing method on top of the AGS:
\begin{enumerate}
	\item Construct the actual AGS via a platform-specific adapter.
	\item Filter the constructed AGS based on user-defined rules.
	\item Check for equality among expected and actual.
	\item If equality succeeds, the GUI is considered OK.
	\item If equality fails, extract the elements to be analyzed.
	\item Identify the elements to see \emph{how} the GUI changed.
	\item Check for inequality to see \emph{what} in the GUI changed.
\end{enumerate}

Compared to existing approaches---be it pixel- or tree-based---the AGS is more flexible, highly customizable and provides better diagnostic information.
Also, to the best of our knowledge, there is no declarative proof system to formalize the act and results of two GUI states being compared.

Next, we explain the details behind the AGS before we conduct an empirical evaluation of our implementation for web-based GUIs.

\section{Abstract GUI State}
\label{sec:ags}

We start by introducing the syntax of an abstract GUI state~(AGS),
before we introduce the details of the equality ($\equality{}{}$), inequality ($\inequality{}{}$) and fairly similar ($\maintainedBy{}{}$) relations.

\subsection{Syntax}
\label{subsec:syntax}

The AGS is a tree-like data structure to provide
for a structure-preserving representation of a GUI.
Its definition is as follows:

\begin{definition}[Abstract GUI State, AGS]
\label{def:ags}

  \bda{rcll}
  G  & ::= & Es                        & \mbox{GUI state}\\
  Es & ::= & [] \mid E : Es            & \mbox{Elements}\\
  E  & ::= & (As, Es)                  & \mbox{Element}\\
  As & ::= & \{\} | \{A\} | As \cup As & \mbox{Attributes}\\
  A  & ::= & (K, V)                    & \mbox{Attribute}\\
  K  & ::= & (\text{string})           & \mbox{Key}\\
  V  & ::= & (\text{string})           & \mbox{Value}\\
  \eda

\end{definition}
A GUI state consists of a list of elements.
We adopt Haskell's notation for lists and write $x:xs$ to denote
a list with head $x$ and tail $xs$.
We assume that the list of objects $[x_1 \dots ,x_n]$
is a shorthand for $x_1 : \dots : x_n :[]$, where $[]$ denotes the empty list.
Each element is represented by a set of attributes and a list of child elements.
An attribute is a simple key-value pair.
For simplicity, we use strings to represent their content.

The AGS representation is derived from a concrete GUI state as shown in the previous section.
Suppose we have a browser with a single opened tab that displays the following HTML:
\begin{verbatim}
<html lang="en">
  <head></head>
  <body>
    <button name="foo">bar</button>
  </body>
</html>
\end{verbatim}
The HTML page can be represented in terms of AGS as follows:
\bda{lll}
G & = & [ (\{ \attribute{lang}{en}, \attribute{type}{html} \}, \\
  &   & \hspace{3pt} [( \{ \attribute{type}{head} \}, []), \\
  &   & \hspace{6pt} ( \{ \attribute{type}{body} \}, [E]) ] ) ]
\eda
Where $E$ represents the \texttt{<button>} element:
\bda{lll}
E & = & (\{ \attribute{name}{foo}, \\
  &   & \hs \attribute{text}{bar}, \\
  &   & \hs \attribute{type}{button} \}, [])
\eda
At this point, the AGS looks like a simple mapping of HTML to another structure.
Of course, this is not sufficient for visual testing because in the case of web-based GUIs, the appearance is influenced by more than just HTML.
While HTML mostly defines structural information, CSS typically provides style information.
Furthermore, JavaScript may be used to modify the DOM.
Our experiments show that such (visual) information can be obtained via a platform-specific adapter from the actual GUI to our AGS.
For example, we may want to include the browser tab as a parent element with attributes such as \attribute{title}{browser tab title}.
Further attributes (keys) are
\texttt{background-color}, \texttt{x}, \texttt{y}, \texttt{height}, \texttt{width} etc.
For brevity, we ignore these attributes in the examples we consider in this section.

A \emph{filter} function $\filter$ may remove
elements (with or without children) or attributes from an AGS instance $G$
such that they are ignored in subsequent steps,
but keeps the hierarchical structure of $G$ intact.

Based on this abstraction, we introduce a highly-automated method for visual testing of web-based GUIs.
In a first step, we establish the notions of \emph{equality}
and \emph{inequality} among two AGSes.

\subsection{Equality and Inequality}
\label{subsec:inequality}

We write $\equality{G}{G'}$ to denote equality and
$\inequality{G}{G'}$ to denote inequality among two AGSes $G$ and $G'$.
Defining inequality as the negation of equality is not sufficient.
If equal, no changes occurred;
otherwise, there have been some changes and in this situation we wish to provide diagnostic information to the user to track down the reasons for these changes.
Hence, we specify an $\inequality{}{}$ relation that yields
the precise reason for failure of equality.
$\equality{}{}$ is then simply defined as the negation of $\inequality{}{}$.
To specify $\inequality{}{}$, we make use of declarative inference rules.

\begin{definition}[Inequality]

  \bda{c}
  \rlabel{Es-Ineq} \
  \myirule{ Es = [ E_1, \dots, E_n ]
          \\ Es' = [ E_1', \dots, E_m' ]
          \\ (n \not= m \vee \exists i\in \{1,\dots,{\mathit \max}(n,m)\}. \inequality{E_i}{E_i'})
          }
          {\inequality{Es}{Es'}}
    
  \\
  \\

  \rlabel{As-Ineq} \
  \myirule{As = \{ A_1, \dots, A_n, B_1, \dots, B_m \}
        \\ \,\,\, \cup \, \, \, \{(K_1,V_1),\dots,(K_l, V_p)\}
        \\ As' = \{ A'_1, \dots, A'_n, B'_1, \dots, B'_m \}
        \\ \,\,\,  \cup \,\,\, \{(K'_1,V'_1),\dots,(K'_l, V'_q)\}
  \\ \equality{A_i}{A'_i} \ \mbox{for $i=1,\dots,n$}
  \\ \inequality{B_j}{B'_j} \ \mbox{for $j=1,\dots,m$}
  \\ \{K_1,\dots,K_p\} \cap \{K'_1,\dots,K'_q \} = \{ \}
  \\ ({\mathit \max}(p,q) \leq 1 \vee m \leq 1)
          }
          {\inequality{As}{As'}}
  
	\\
	\\
	
  \ba{cc}
  \rlabel{E-Ineq} \
  \myirule{ \inequality{As}{As'} \\ \vee \\
            \inequality{Es}{Es'}
          }
          {\inequality{(As, Es)}{(As', Es')}}
  
  &
  \rlabel{A-Ineq} \
  \myirule{ K = K' \\ V \not= V' }
          {\inequality{(K,V)}{(K',V')}}          
  \ea              
  \eda
\end{definition}
Inference rules are to be read as follows:
The statements above the bar are the preconditions (premise).
The statement below the bar is the conclusion.
If the premise can be satisfied, we can derive the conclusion.
The premise represents detailed diagnostic information that explains
the reason for inequality.

For each syntactic case, there exists an inference rule.
Rule \rlabel{Es-Ineq} identifies two lists of elements
as inequal, if their size differs or some elements are inequal.
For elements, inequality is due to their attributes
or some child elements, see rule \rlabel{E-Ineq}.
The interesting inequality inference rules are \rlabel{As-Ineq} and \rlabel{A-Ineq}.
Rule \rlabel{A-Ineq} identifies attributes with matching keys but non-matching values.
Rule \rlabel{As-Ineq} identifies all attributes $A_i$ and $A_i'$ that are matching and
attributes $B_j$ and $B_j'$ that are non-matching due to different values.
This rule also collects all attributes that have no counterpart based on their key.

Our inference rules effectively represent Prolog-style Horn clauses.
Hence, to decide inequality we resolve (rewrite)
matching conclusions by the respective premise
until no further resolution steps are possible.
It is easy to see that  inference rules are terminating.

Recall the elements $E$ and $E'$ from Section~\ref{sec:overview}.
$E$ and $E'$ are inequal, and based on the inequality relation
we can derive detailed diagnostic information for the reason(s)
of inequality:
\bda{ll}
\label{arr:ineq-example}
& \inequality{E}{E'}
\\
\rightarrow
& (1) \ \equality{\attribute{id}{login}}{\attribute{id}{login}}
\\
& (2) \ \inequality{\attribute{b-color}{\#0}}{\attribute{b-color}{\#2}}
\\
& (3) \ \inequality{\attribute{text}{Sign in}}{\attribute{text}{Log in}}
\\
& (4) \ \inequality{\attribute{type}{a}}{\attribute{type}{button}}
\\
& (5) \ \{ \texttt{href} \} \cap \{ \texttt{onclick} \} = \{ \}           
\eda
We use \texttt{b-color} as an abbreviation for
\texttt{background-color} and shorten the color code values.
We write $\rightarrow$ to denote the construction of the inequality
proof, where statements \numrange{1}{5} represent the conditions
collected from the leaf nodes in the inequality proof.

Often, simple structural changes immediately lead to inequality.
Consider the AGSes $G = [E, E_2]$ and $G' = [E_2, E']$, where $E, E'$ as above and $E_2$ is some other but unchanged element.
Our inequality test reports $\inequality{E}{E_2}$ and $\inequality{E_2}{E'}$.
Rather, we wish to report $\inequality{E}{E'}$ and $\equality{E_2}{E_2}$.

\subsection{Fairly Similar and Element Identification}
\label{subsec:fairly-similar}

In order to address the above issue,
we have to identify functionally equivalent GUI elements between two AGSes,
which then can be used to derive inequality.
This is also known as the \emph{GUI element identification problem}, coined by McMaster and Memon~\cite{mcmaster09},
that typically occurs in test script repair.
For each GUI element in $G \cup G'$, one has to decide whether it has been \emph{deleted} (present in expected but not in actual),
\emph{created} (present in actual but not in expected) or if it is \emph{maintained} (present in both but possibly modified).
Therefore, we first assume a weaker form of equality.

\begin{definition}[Fairly Similar]
  Let $E$ and $E'$ be two elements.
  We write $\maintainedBy{E}{E'}$ to denote that $E$
  and $E'$ are \emph{fairly similar}.
	We sometimes also say that $E$ is \emph{maintained by} $E'$.
\end{definition}
In contrast to McMaster and Memon,
we don't limit the GUI element identification problem to \enquote{actionable} elements (e.g. buttons),
as we also want to provide diagnostic information for \enquote{non-actionable} elements (e.g. labels).

Depending on the application context, we might only want to consider leaf or other selected elements.
For this purpose, we introduce the notion of extracted elements.

\begin{definition}[Extracted Elements]
  Let $G$ be an AGS.
  We write $\extract$ to denote
  a function that, when applied to $G$, yields a set of elements.
  We refer to $\extract$ as the \emph{element extractor} function
  and for each element in $\extract(G)$ as an \emph{extracted element}.
\end{definition}
Compared to $\filter$, $\extract$ dissolves the hierarchical structure of $G$.

Now, we give a concrete specification for the GUI element identification problem in terms of AGS.

\begin{definition}[Element Identification]
  Let $G$ and $G'$ be the expected and actual AGSes.
  Then, we define
  \begin{align*}
    D & = \{ E_d \mid E_d \in \extract(G) \land \nexists E' \in \extract(G') : \maintainedBy{E'}{E_d} \}
    \\
    C & = \{ E_c \mid E_c \in \extract(G') \land \nexists E \in \extract(G) : \maintainedBy{E}{E_c} \}
    \\
    M & = \{ (E, E') \mid E \in \extract(G) \land \exists E' \in \extract(G') : \maintainedBy{E}{E'} \}
  \end{align*}        
  where we refer to $D$ as the set of \emph{deleted} elements,
  $C$ as the set of \emph{created} elements and
  $M$ as the set of \emph{maintained} pairs.

  We assume a function $\idente$
  that computes and returns these three sets.
  For $G$ and $G'$,
  we find $\idente(G, G') = (D, C, M)$,
  where $E,C, M$ are defined as above.
\end{definition}

Based on equality, inequality and element identification,
we can define the individual AGS execution steps (\numrange{3}{7})
motivated in the overview Section~\ref{sec:overview} as follows:

\begin{definition}[AGS Executor for Golden Master Testing]
  Let $G$ and $G'$ be the expected and actual AGSes.
\begin{algorithmic}[1]
	\If{$\equality{G}{G'}$ derivable}
		\State report GUI is OK
	\Else
		\State $R \gets \emptyset$ \Comment{initialize inequality results}
		\State $(D, C, M) \gets \idente(G, G')$
		\For{each $(E, E') \in M$}
			\State $R \gets R \cup \inequality{E}{E'}$ \Comment{collect inequality result}
		\EndFor
		\State report diagnostic information for $D, C, R$
	\EndIf
\end{algorithmic}
\end{definition}

\subsection{Instances}
\label{subsec:instances}

In the following,
we consider three specific instances of fairly similar
to explore the design space,
and which we will also compare in our empirical evaluation.
We first define the extractor function shared by all instances:
\bda{lcl}
\extract([])                & = & \{ \} \\
\extract([E_1, \dots, E_n]) & = & \extract(E_1) \cup \dots \cup \extract(E_n) \\
\extract((As, Es))          & = & \{ (As, []) \} \cup \extract(Es)
\eda
For each element, we extract its set of attributes.
We consider parent and child elements, but all extracted elements are flattened
by leaving their children to be empty.
Based on this extractor function, the fairly similar relation only needs to
compare sets of attributes.

Let $I$ be a set of attribute keys and $As$ be a set of attributes.
We define $\projKV{As}{I} = \{ (K, V) \mid (K, V) \in As \wedge K \in I \}$
and $\projK{As}{I} = \{ K \mid (K, V) \in As \wedge K \in I \}$.

We assume two sets $I_s$ and $I_w$ of attribute keys,
where $I_s$ is referred to as the \emph{strong} identifying set of keys
and $I_w$ is referred to as the \emph{weak} identifying set of keys.
We assume that $I_s \cap I_w = \emptyset$.

\begin{definition}[Fairly Similar by Strong and Weak Keys]
  We define an instance of the fairly similar relation
  referred to as \emph{fairly similar by strong and weak keys} as follows:
	$\maintainedByI{(As, [])}{(As', [])}$ holds if
	\begin{inparaenum}[(1)]
		\item $\equality{\projKV{As}{I_s}}{\projKV{As'}{I_s}}$ and
		\item $\projK{As}{I_w} = \projK{As}{I_w}$.
	\end{inparaenum}
\end{definition}
The first condition says that for all strong attribute keys
the values must remain unchanged.
The second condition says that for all weak attribute keys
we only check if the key is present and ignore the value.
Attribute keys not appearing in $I_s \cup I_w$ are ignored.

It is easy to see that the above definition
of fairly similar satisfies the laws of an equivalence relation,
although this is not required by an instance.

Sometimes also strong attribute keys change,
hence, demanding equality among them doesn't work.
For this, we require a less strict relation.

\begin{definition}[Fairly Similar by Key Tests]
  We define an instance of the fairly similar relation
  referred to as \emph{fairly similar by key tests} as follows:
	$\maintainedByII{(As, [])}{(As', [])}$ holds if
	\begin{inparaenum}[(1)]
		\item $\forall (K, V_1) \in \projKV{As}{I_s},
				(K, V_2) \in \projKV{As'}{I_s} :
        f_K(V_1, V_2)$ and
		\item $\projK{As}{I_w} = \projK{As}{I_w}$.
	\end{inparaenum}
\end{definition}
For each key $K$ in $I_s$ we assume a function $f_K$ that takes two values $V_1, V_2$
(expected and actual) belonging to key $K$ and yields either true or false.
That is, we only check if they agree based on the fairly similar test $f_K$:
\bda{ll}
f_K =
\begin{cases}
	\text{true}  & \text{if } \simjw(V_1, V_2) \geq t \\
	\text{false} & \text{otherwise}
\end{cases}
\eda
While $\simjw$ refers to the Jaro-Winkler string similarity metric.
If $V_1$ and $V_2$ are similar up to a threshold $t \in [0, 1]$,
$f_K$ yields true, otherwise false.

Note that we currently limit value types to string,
but the more types we introduce, the more precise fairly similar can be.
Consider, for example, floats for the element positions,
then we could compute the distance between expected and actual in two-dimensional space.

The third instance of fairly similar we consider computes a match score
for each extracted element pair $(E, E')$.
Let $(As,[])$ and $(As',[])$ be two extracted elements.
We define a function to check if the values for some
key $K$ found in $As$ and in $As'$ are equal:
\bda{ll}
\eqv(As, As', K) =
\begin{cases}
	\text{1} & \text{if } (K, V) \in As \land (K, V') \in As : V = V' \\
	\text{0} & \text{otherwise}
\end{cases}
\eda
The \enquote{otherwise} case applies if the values for $K$ are different,
or if $K$ is missing in an attribute set.

If the average number of equal values for keys in $I_s \cup I_w$ is above a certain threshold $u \in [0, 1]$,
the element pair shall be fairly similar.

\begin{definition}[Fairly Similar by Element Matching]
  We define an instance of the fairly similar relation
  referred to as \emph{fairly similar by element matching} as follows:
	$\maintainedByIII{(As, [])}{(As', [])}$ holds if
	$(\frac{1}{n} \cdot \sum_{K \in I_s \cup I_w} \eqv(As, As', K)) \geq u$,
	where $n = \vert I_s \cup I_w \vert$ denotes the combined number of attributes.
\end{definition}
Note that the \emph{overall} best matching extracted element pair $(E, E')$ determines the final assignment.

We recap our example from Section~\ref{subsec:inequality}, where $G = [E, E_2]$ and $G' = [E_2, E']$.
We set $I_s = \{ \str{id} \}, I_w = \{ \str{text}, \str{type} \}$ and $t = 0.9, u = 0.3$.

At first, we cannot derive $\equality{G}{G'}$
because $\inequality{E}{E_2}$ and $\inequality{E_2}{E'}$.
We then attempt to identify the correct element pairs to be used for comparison.
In the course of this, we realize that $E_2 \in G$ is the same as $E_2 \in G'$,
so all instances of fairly similar resolve $\maintainedBy{E_2}{E_2}$.

When it comes to $E, E'$,
fairly similar by strong and weak keys holds as the strong key \texttt{id} remains unchanged
and the weak keys \texttt{text} and \texttt{type} are still present in $As, As'$.

Fairly similar by key tests also holds because the Jaro-Winkler similarity for the unchanged \texttt{id} is 1.0
(i.e. greater than $t$),
and the weak keys are tested as above.

Fairly similar by element matching holds as well:
$\eqv(As, As', \str{id})$ yields 1 and 0 in the other cases.
This results in a match score of
$\frac{1}{3} \cdot (1 + 0 + 0) = \frac{1}{3} = 0.\overline{3}$,
which is greater than $u$.

Thus, we conclude $D = C = \emptyset$ and $M = \{ (E, E'), (E_2, E_2) \}$,
so we report $\inequality{E}{E'}$ and $\equality{E_2}{E_2}$.

Note that certain changes to strong keys may quickly lead
to wrongly identified deleted and created elements.
Suppose the \texttt{id} of \str{login} becomes \str{signin},
then $\equality{\projKV{As}{I_s}}{\projKV{As'}{I_s}}$ fails.
And since the Jaro-Winkler similarity (0.59) drops below $t$,
$f_K$ fails too.
Now that all values for the keys in $I_s \cup I_w$ are different,
the match score is 0 and $\maintainedByIII{}{}$ also doesn't hold anymore.
As a consequence, we conclude $D = \{ E \}, C = \{ E' \}, M = \{ (E_2, E_2) \}$.

Ergo, the quality of fairly similar strongly depends on determining adequate strong and weak keys as well as appropriate thresholds.
For further discussion, see our upcoming experiments.

\section{Empirical Evaluation}
\label{sec:eval}

We have implemented a platform-independent test framework called recheck~\cite{retest19a},
which contains the AGS-centric components.
It follows the principles of golden master testing and offers mechanisms to
create and maintain visual tests.
recheck is extensible and provides interfaces to support specific GUI technologies.
We used this as a foundation to create an AGS adapter for web-based GUIs named recheck-web~\cite{retest19b},
which operates on top of the Selenium WebDriver API~\cite{selenium16}.
We currently provide a Java-based SDK that seamlessly integrates in existing Selenium tests.

To empirically evaluate the implementation of our approach,
we created an open-source benchmark for visual testing of web-based GUIs~\cite{kraus19}.
The benchmark suite consists of offline versions from 20 of the most popular websites~\cite{wikipedia19}.
We downloaded the pages to have better control and
to minimize external influences during the experiments.
Every web page was modified to simulate typical GUI changes.
The benchmark is implemented as a Maven-based Java project and offers several utilities.
For example, to measure execution times and
to count the number of GUI elements.

Using this benchmark,
we want to answer the following three research questions~(RQs):
\begin{description}
	\item[RQ$_1$] How long does it take for each fairly similar instance to compare two GUI states?
	\item[RQ$_2$] How does each fairly similar instance perform when it comes to element identification?
	\item[RQ$_3$] Using the best-performing fairly similar instance, what types of GUI changes can we detected?
\end{description}
In the context of our study,
RQ$_1$ and RQ$_2$ are directed towards quantitatively measuring the performance of our implementation,
also to see if it is a practical tool for visual testing.
RQ$_3$, however, aims at qualitatively measuring the effectiveness of our approach.

As an extension of RQ$_3$,
we additionally compare our implementation against Applitools using
their own demo application from Figure~\ref{fig:login-all},
where we highlight some findings from the error reports of both tools.
The assumption we made is that if we manage to compete with Applitools
in an environment optimized for their own tool,
it would be a strong indicator for the
usefulness of our implementation.

Next, we outline our evaluation setup in Section~\ref{subsec:eval-setup}.
We then present and discuss the results in Section~\ref{subsec:results},
followed by limitations as well as possible threats to validity in Section~\ref{subsec:limits-threats}.

\subsection{Evaluation Setup}
\label{subsec:eval-setup}

As mentioned above, we utilized our own open-source benchmark.
It is available as a Git repository on GitHub,
where we also published all relevant data that was collected during our experiments.
The Git repository can be cloned to reproduce the evaluation,
either locally or on Travis CI,
for which we provide a build configuration.
(Please consider the reproducibility notes from the \texttt{README.md} file.)

All tests have been executed on a virtual machine running Ubuntu 18.04.3,
with a \SI{2.3}{\giga\hertz} Intel Xeon Gold 6140 CPU and \SI{4}{\giga\byte} of RAM.
To honor the importance of cross-browser and \=/platform visual testing,
we ran the tests on two different browsers with different resolutions:
Chromium (79.0.3945.79 via ChromeDriver 79.0.3945.79) on 1080p and
Firefox (72.0.1 via GeckoDriver 0.26.0) on 720p.

\begin{table}[b]
	\scriptsize
	\centering
	\caption{Adopted GUI changes taxonomy~\cite{moran18}.}
	\label{tab:gui-changes-taxonomy}
	\begin{tabular}{ll} \toprule
		Change category &
		Description \\ \midrule
		\multirow{3}{*}{1\quad Text change}
		&
		1.1\quad Text content change \\
		&
		1.2\quad Font change \\
		&
		1.3\quad Font color change \\ \addlinespace
		\multirow{2}{*}{2\quad Layout change}
		&
		2.1\quad Horizontal or vertical element translation \\
		&
		2.2\quad Horizontal or vertical element size change \\ \addlinespace
		\multirow{3}{*}{3\quad Resource change}
		&
		3.1\quad Deleted element \\
		&
		3.2\quad Created element \\
		&
		3.3\quad Element type change \\ \bottomrule
	\end{tabular}
\end{table}

To simulate typical changes,
we reused the GUI changes taxonomy recently introduced by Moran et al.~\cite{moran18} and
adopted it as shown in Table~\ref{tab:gui-changes-taxonomy}.
We first created a golden master for each original web page.
To ensure that the pages are in a steady state,
we added a fixed page load wait of \SI{3}{\second}.
We also filtered out dynamic elements such as carousels and
various (mostly invisible) attributes by adding them to the \texttt{recheck.ignore} file.
This text file represents the user interface for the AGS filter function $\filter$ by
supporting a set of rules to ignore element and attribute differences.
In addition, pixel diffs up to 25 are ignored too.

Finally, we manually introduced every GUI change type from the taxonomy,
i.e., 8 per web page and $8 \cdot 20 = 160$ in total.
The changes have been applied by multiple people to foster diversity.

To answer RQ$_1$,
we ran each test 10 times against the modified page
using the 3 fairly similar instances from Section~\ref{subsec:instances}.
For each browser (2) and website (20),
this resulted in $10 \cdot 2 \cdot 20 = 400$ test executions per instance.
This allowed us to obtain a realistic estimate for the average execution times.

As part of our page modifications,
GUI elements have been deleted, created and maintained.
The affected elements should be identified correctly between the expected and actual GUI states.
For RQ$_2$, we again compared our 3 fairly similar instances
to determine the numbers of deleted, created and maintained elements.
The better the element identification, the better the diagnostic information.

Based on the results of RQ$_1$ and RQ$_2$,
we selected the best performing fairly similar instance to answer RQ$_3$.
Multiple people examined the test reports in detail and
categorized the resulting diagnostic information.
Ideally, a single difference is reported for each introduced GUI change,
i.e., only true positives~(TPs).
If a change isn't reported,
we interpreted this as a false negative~(FN);
if an additional change is reported,
we interpreted this as a false positive~(FP).
True negatives~(TNs) are irrelevant for this experiment
as this would only denote the number of unchanged elements.

Note that a GUI change may affect other elements too and
the user must decide whether this information is desired or not.
Given the amount of pages and browsers of the benchmark,
this is difficult to do in a consistent and objective manner.
Therefore, we treat such diffs as TPs because they reflect actual changes and
can be easily ignored using the AGS filter mechanism.

For the fairly similar instances, we set
$I_s = \{ \str{id}, \str{path} \}$ and
$I_w = \{ \str{type}, \str{x}, \str{y}, \str{width}, \str{height} \}$.
In addition, we included the keys
$\texttt{class}, \texttt{id}, \texttt{name}, \texttt{text}$
for fairly similar by element matching.
Thresholds are set to $t = 0.9, u = 0.3$.

\subsection{Experimental Results \& Discussion}
\label{subsec:results}

The results for RQ$_1$ are shown in Figure~\ref{fig:rq1-plot}.
The execution times of $\maintainedByI{}{}$ (by strong and weak keys)
were between \SIrange{1435}{61091}{\milli\second},
with an average of \SI{6517.72}{\milli\second}.
$\maintainedByII{}{}$ (by key tests) was between \SIrange{1513}{56617}{\milli\second}
with \SI{6350.37}{\milli\second} on average, and
$\maintainedByIII{}{}$ (by element matching) was between \SIrange{1267}{57188}{\milli\second}
with \SI{5822.29}{\milli\second} on average.
That is, $\maintainedByIII{}{}$ was faster than $\maintainedByII{}{}$,
and $\maintainedByII{}{}$ faster than $\maintainedByI{}{}$.

Figure~\ref{fig:rq2-plot} shows the results for RQ$_2$.
$\maintainedByI{}{}$ reported \num{2433} differences
(586 deleted, \num{1472} created, 375 maintained) and
$\maintainedByII{}{}$ reported \num{21440} differences
(\num{9946} deleted, \num{10881} created, \num{613} maintained).
$\maintainedByIII{}{}$ performed best by reporting 762 differences
(\num{50} deleted, \num{43} created, \num{669} maintained).
Remember, originally we only introduced 160 GUI changes per browser (i.e. 320 combined),
although this doesn't account the thereby affected elements.

As pointed out in Section~\ref{subsec:instances},
wrong element identifications can increase the number of
deleted and created elements, and thus the amount of differences.
The reason why $\maintainedByII{}{}$ performed so bad is
because the fairly similar test $f_K$ wasn't strict enough.
This messed up the element identification in a way
such that incorrect pairs have been created,
which lead to many other differences due to mixed up elements.

\begin{figure}
	\subfloat[\label{fig:rq1-plot} Time measurements in milliseconds.]{
		\includegraphics[width=0.96\columnwidth]{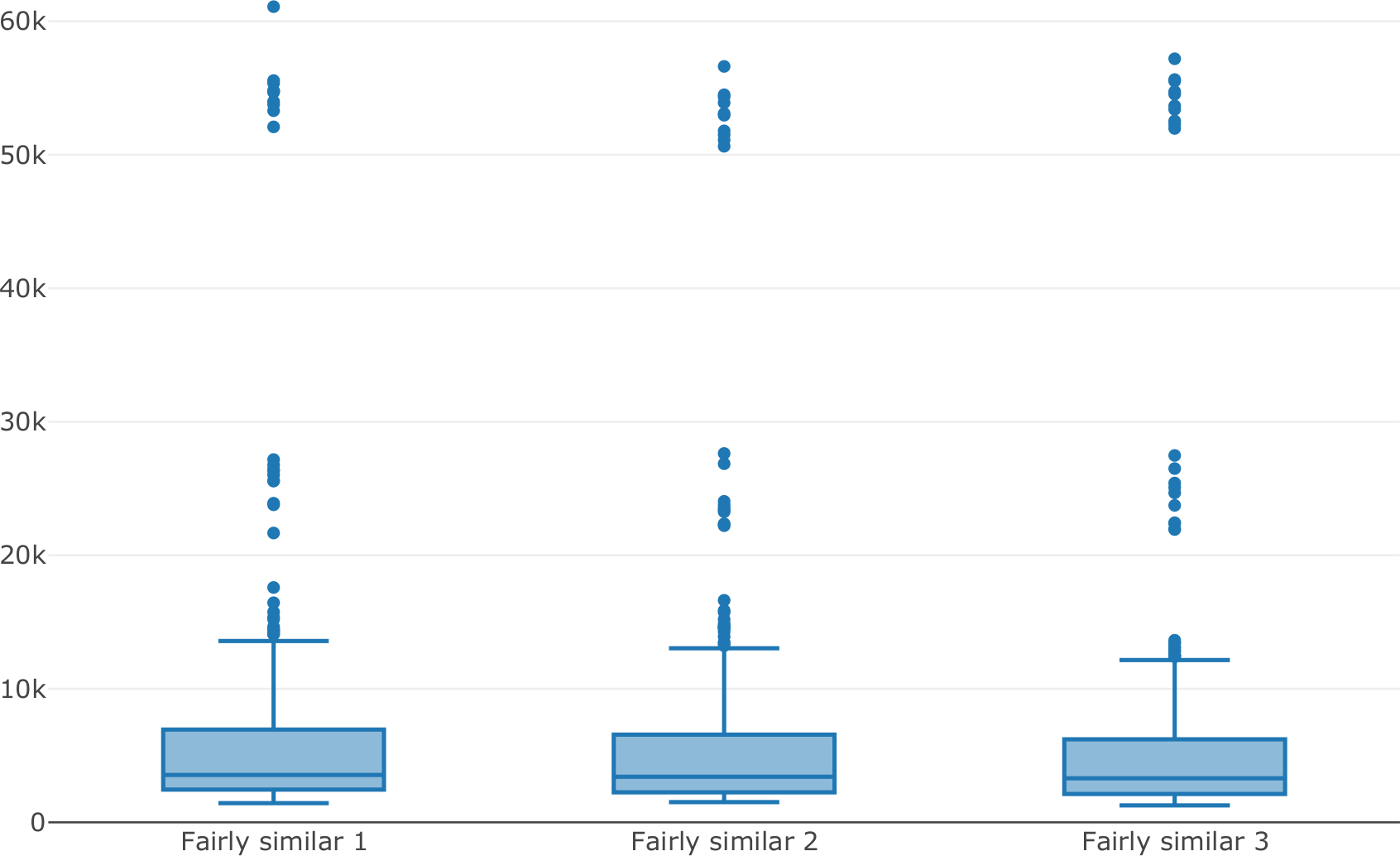}}
	\\
	\subfloat[\label{fig:rq2-plot} Deleted, created and maintained elements.]{
		\includegraphics[width=\columnwidth]{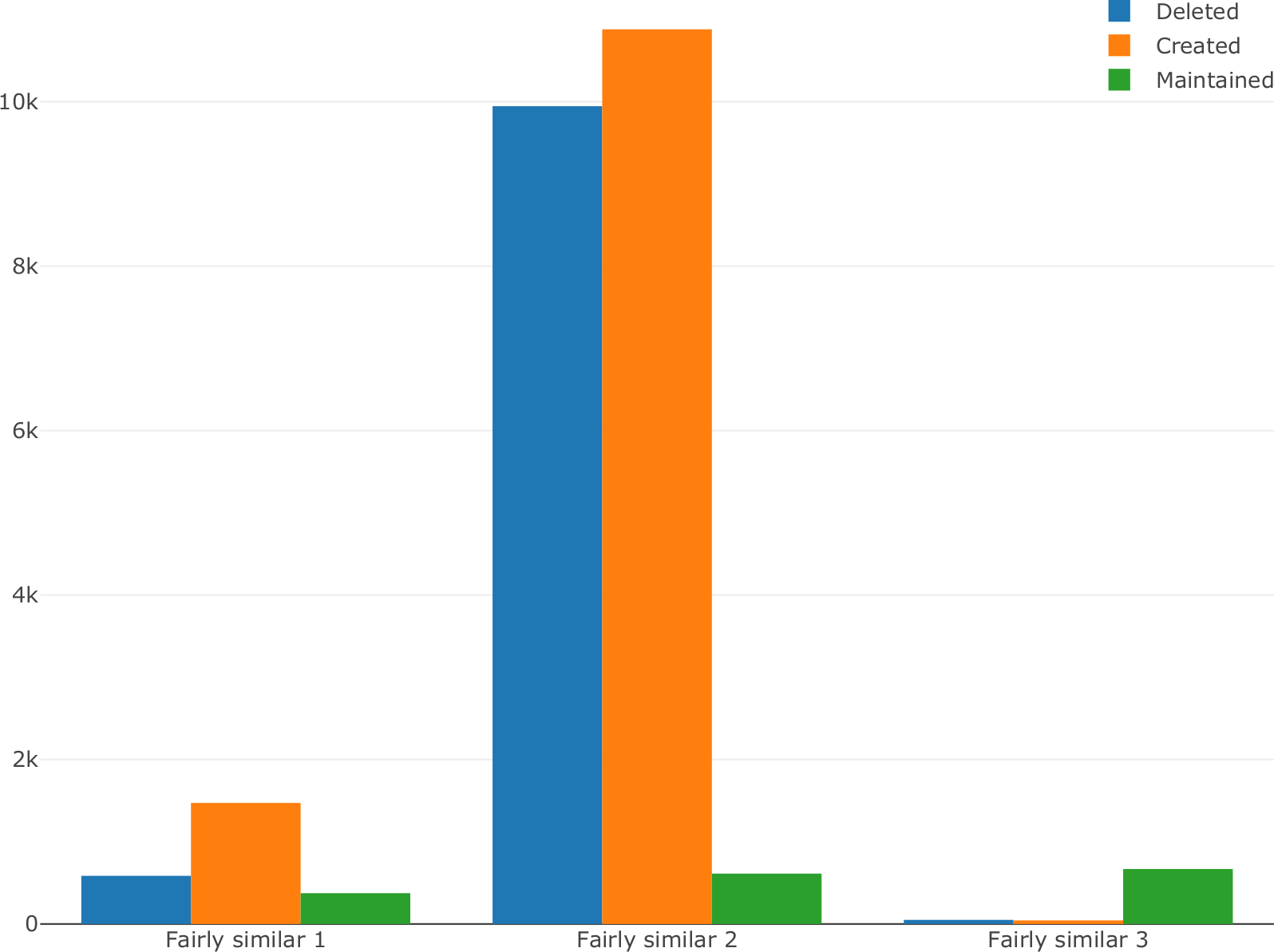}}
	\caption{Results for RQ$_1$ and RQ$_2$ by fairly similar instance.}
\end{figure}

Because $\maintainedByIII{}{}$ overall performed best during the previous experiments,
we selected it to continue with RQ$_3$.
An overview of the results can be found in Table~\ref{tab:rq-1-results}.
First, it should be noticed that the number of elements (and TPs, FNs, FPs)
sometimes varies between Chromium and Firefox.
This can be caused by different DOMs or
due to responsive design,
which adjusts the content of a web page
with respect to the given screen size and resolution.
(Remember, for Chromium we used 1080p and for Firefox 720p.)
In total, we tested our implementation against
\num{19588} elements in Chromium and \num{19672} elements in Firefox.

\begin{table}
	\scriptsize
	\centering
	\caption{Overview of TPs, FNs, FPs in terms of GUI change detection performance from RQ$_3$.}
	\label{tab:rq-1-results}
	\begin{tabular}{l rr rrr rrr rrr rrr} \toprule
		&
		\multicolumn{4}{c}{Chromium (1080p)} &
		\multicolumn{4}{c}{Firefox (720p)} \\
		Website &
		\# elements &
		TP &
		FN &
		FP &
		\# elements &
		TP &
		FN &
		FP \\ \midrule
		360.cn &
		\num{1250} &
		26 &
		0 &
		1 &
		\num{1250} &
		31 &
		0 &
		1 \\
		alipay.com &
		\num{95} &
		11 &
		0 &
		0 &
		\num{95} &
		11 &
		0 &
		0 \\
		apple.com &
		\num{633} &
		11 &
		1 &
		2 &
		\num{633} &
		11 &
		1 &
		3 \\
		baidu.com &
		\num{152} &
		9 &
		2 &
		2 &
		\num{152} &
		9 &
		2 &
		2 \\
		bbc.com &
		\num{1245} &
		11 &
		0 &
		1 &
		\num{1289} &
		0 &
		8 &
		2 \\
		blogspot.com &
		\num{291} &
		12 &
		0 &
		3 &
		\num{291} &
		11 &
		0 &
		3 \\
		csdn.net &
		\num{1992} &
		22 &
		0 &
		20 &
		\num{1992} &
		22 &
		0 &
		20 \\
		ebay.com &
		\num{2176} &
		49 &
		0 &
		27 &
		\num{2167} &
		48 &
		0 &
		26 \\
		facebook.com &
		\num{611} &
		10 &
		0 &
		8 &
		\num{611} &
		10 &
		0 &
		8 \\
		github.com &
		\num{972} &
		20 &
		0 &
		2 &
		\num{972} &
		24 &
		0 &
		4 \\
		google.com &
		\num{227} &
		14 &
		1 &
		0 &
		\num{227} &
		14 &
		1 &
		0 \\
		jd.com &
		\num{266} &
		0 &
		8 &
		1 &
		\num{266} &
		0 &
		8 &
		2 \\
		linkedin.com &
		\num{644} &
		7 &
		2 &
		0 &
		\num{652} &
		7 &
		2 &
		0 \\
		live.com &
		\num{445} &
		23 &
		1 &
		1 &
		\num{445} &
		23 &
		1 &
		3 \\
		soso.com &
		\num{25} &
		19 &
		0 &
		2 &
		\num{25} &
		19 &
		0 &
		2 \\
		stackoverflow.com &
		\num{837} &
		16 &
		1 &
		0 &
		\num{837} &
		16 &
		1 &
		0 \\
		twitter.com &
		\num{667} &
		15 &
		1 &
		0 &
		\num{673} &
		15 &
		1 &
		0 \\
		vk.com &
		\num{278} &
		12 &
		0 &
		0 &
		\num{280} &
		12 &
		0 &
		0 \\
		wikipedia.org &
		\num{942} &
		15 &
		1 &
		0 &
		\num{942} &
		15 &
		1 &
		0 \\
		youtube.com &
		\num{5840} &
		3 &
		4 &
		5 &
		\num{5873} &
		3 &
		4 &
		5 \\ \addlinespace
		Total &
		\num{19588} &
		305 &
		22 &
		75 &
		\num{19672} &
		301 &
		30 &
		81 \\ \bottomrule
	\end{tabular}
\end{table}

Based on the amount of TPs, FNs and FPs, we can determine
how many selected items are relevant,
and how many relevant items are selected:
\bda{ll}
\text{Precision} = \frac{TP}{TP + FP} & \text{Recall} = \frac{TP}{TP + FN}
\eda
For Chromium, we achieved a precision of \SI{80.26}{\percent} and recall of \SI{93.27}{\percent},
in Firefox a precision of \SI{78.80}{\percent} and recall of \SI{90.94}{\percent}.
Overall, this results in a precision of \SI{79.53}{\percent} and recall of \SI{92.10}{\percent}.

An observation we made is that
on some pages (e.g. \url{bbc.com} in Firefox) the algorithm
wasn't able to correctly identify the root element.
This resulted in the entire page being treated as deleted and inserted.
We interpreted this as 8 FNs (missing original changes) and
2 FPs (wrongly deleted/inserted root).

We also observed that some FNs were caused by filtered elements.
This is likely because we involved multiple people in our experiments,
and those that introduced the changes didn't know which elements
have been considered unimportant by the person who has created the \texttt{recheck.ignore}.

\begin{figure}
	\includegraphics[width=0.95\columnwidth]{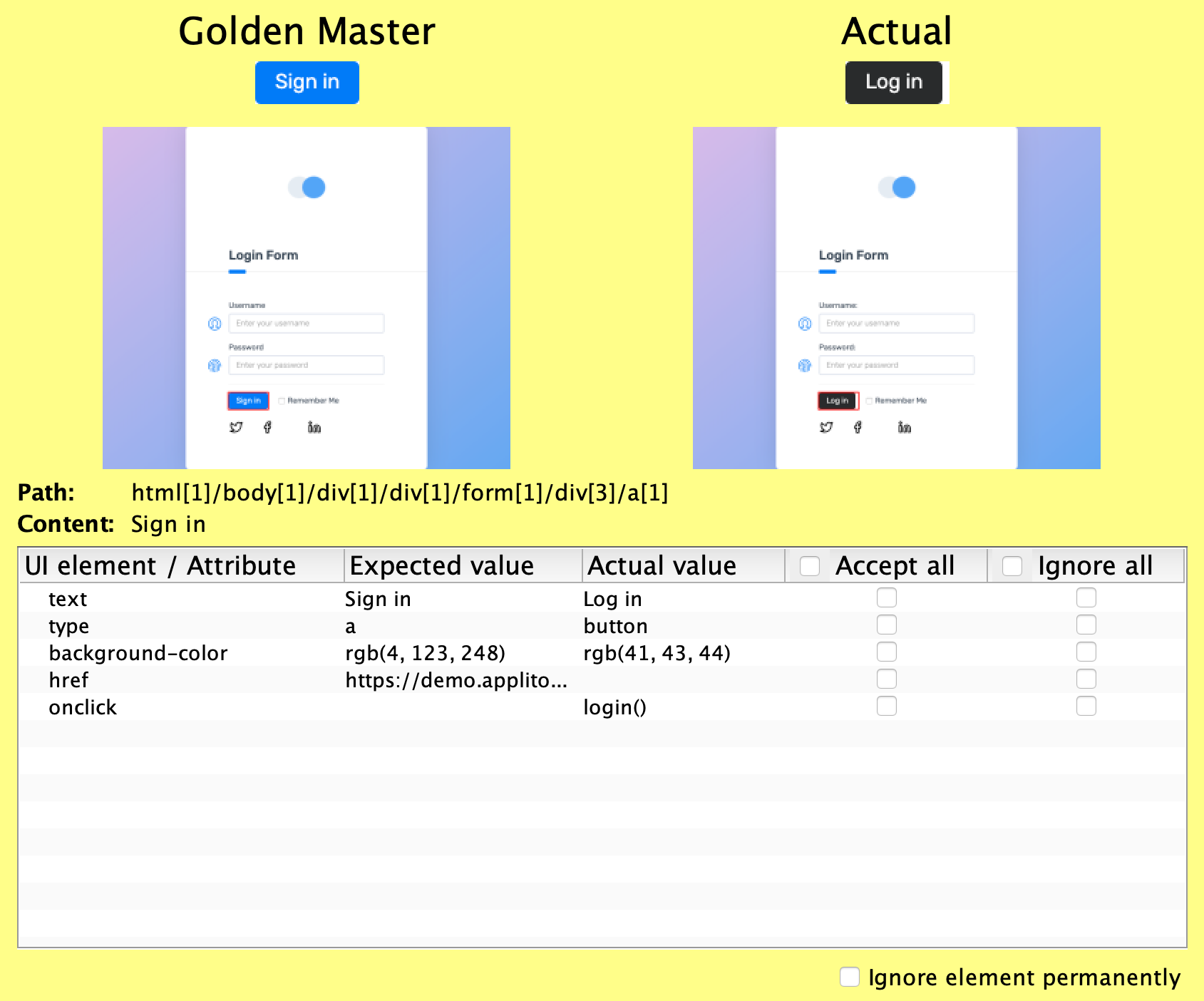}
	\caption{Test report excerpt for changes to the \enquote{Sign in} button from Applitools' demo.}
	\label{fig:example-test-report}
\end{figure}

As an extension of RQ$_3$,
we also compared ourselves against Applitools
using their own demo application,
where we wanted to highlight some findings from the error reports.
recheck-web managed to detect all introduced GUI changes,
whereas Applitools wasn't able to expose the element type change (3.3).
This is because Applitools exclusively tracks visible information.

Another thing we noticed is that Applitools'
screenshot was covered with many diffs,
also in places where the page wasn't modified.
In contrast, Figure~\ref{fig:example-test-report}
shows the changes to the former \enquote{Sign in} button
displayed in our GUI~\cite{retest19d}.
To assist the user, we include a screenshot
of the affected element and the entire screen,
in which the element is highlighted.
Below, the attribute changes are listed:
\begin{inparaenum}[(i)]
	\item \texttt{text},
	\item \texttt{type},
	\item \texttt{background-color},
	\item \texttt{href} and
	\item \texttt{onclick}.
\end{inparaenum}
This corresponds exactly to our inequality example on p.~\pageref{arr:ineq-example}.
A change can be either accepted, which updates the corresponding golden master,
or ignored, which updates the \texttt{recheck.ignore}.
If a change is unintended, the test report serves as documentation.

As most changes can be broken down into inequal attributes,
we are able to further automate test maintenance.
Suppose we have several test cases that all pass the login screen.
If we apply the \texttt{text} change from above,
all of these tests fail with differences.
However, for each test case we have the very same failure case \rlabel{A-Ineq},
where $K = K' = \str{text}$, but $V = \str{Sign in}$ and $V' = \str{Log in}$.
If the user accepts this change once,
we can suggest to apply the decision
to all the same changes---even across browsers and platforms---since we are able
to identify the button element also elsewhere.

For further details such as screenshots and raw data,
please refer to the benchmark repository.

\subsection{Limitations \& Threats to Validity}
\label{subsec:limits-threats}

While the experimental results demonstrate the effectiveness and efficiency of our approach,
it relies on an accurate adapter.
This mostly influences the performance of our method for visual testing and
requires platform-specific knowledge to take characteristics of the GUI technology into account.
We observed that large, complicated DOMs (e.g. \url{youtube.com})
can make it difficult to understand the reported changes.
This is an advantage of image comparison because it works out of the box for almost every platform.
But there is a serious trade-off in the quality of test results~\cite{alegroth18}.
And without metadata, pixel-based approaches can hardly provide useful diagnostic information.

recheck and recheck-web are currently designed to be used
in hand-crafted test scripts,
which is the de-facto standard in the industry~\cite{gao16}.
An extended prototype could leverage exploration strategies
to compare two versions of a GUI autonomously, similar to GCat~\cite{moran18}.
This would require no adaptions in our implementation
as a corresponding crawler could easily leverage our existing APIs
to automatically integrate visual checkpoints.

Another limitation is that layout errors like overflows
can be detected, but they aren't classified properly.
Since the AGS contains semantic information about the GUI,
it can basically be used to statically detect layout failures too,
and report them to tools such as Viser~\cite{althomali19}.

In terms of \emph{external threats} to our experimental results,
the set of websites may not generalize well.
Although we carefully selected 20 examples from the most popular websites,
other pages possibly use special GUI libraries that lead to tricky situations.
For instance, randomly generated attributes would first create a lot of differences,
but the user can easily ignore these globally
to overcome this problem---just like we did in the evaluation.
Also frequently changing content can cause flaky tests,
but pixel-based approaches suffer from this limitation too.
Thus, the user is responsible to ensure a stable test environment.
With the AGS filter mechanism, this is less challenging and
compared to ignore regions when using image comparison,
not all information within the region is lost.

\emph{Internal validity threats} may be due to our set of RQs.
RQ$_1$ was dedicated to the execution times of our implementation.
A study among more tools and websites could have revealed performance bottlenecks.
For RQ$_2$, we tested our instances of fairly similar to see
if we are able to reliably identify elements and,
therefore, provide useful diagnostic information.
While this yielded promising results,
we might want to leverage advantages
like those by Grechanik et al.~\cite{grechanik18}
to see if we can further improve.
As part of RQ$_3$, we selected specific GUI changes from a modern taxonomy
we find important in the context of visual testing.
Applying these changes and categorizing the results
into TPs, FNs and FPs is to some degree a subjective task.

Overall, we believe that our RQs showed
the effectiveness and efficiency of our approach,
both qualitatively as well as quantitatively,
on the basis of comprehensible and transparent experiments.

\section{Related Work}
\label{sec:related-work}

There is a large body of work that deals with
GUI testing in general~\cite{ermuth16, gao15, mao17, moreira17, saddler17, song17}.
However, most of these contributions are dedicated to testing applications via the GUI,
rather than testing the GUI itself.
The same goes for many \enquote{visual GUI testing} approaches,
where CV and other techniques are used to interact with the SUT through the GUI~\cite{alegroth13, borjesson12, yeh09}.
We specifically want to look at methods for checking visual correctness in the following.

In the area of manual testing, some tools~\cite{luo18, quirktools19, testsize19}
allow to visually inspect the GUI using a range of common screen sizes and resolutions.
While this is helpful for sanity checks or exploratory testing,
manually checking for visual correctness is time consuming,
inconsistent and prone to human errors~\cite{althomali19}.

Various automated approaches~\cite{halle16, panchekha18, zaiats19} are based on formal specifications,
which define desired GUI properties.
A well-known industrial example for this is the Galen Framework~\cite{galenframework19}.
The user creates the specification using a DSL that
describes certain GUI element properties (e.g. width) and how GUI elements relate to each other.
But writing and maintaining specifications can be tedious and requires manual effort.
More importantly, a specification can only protect against expected changes~\cite{slatkin13}.
Changes which are not covered by the specification are implicitly allowed as they do not lead to test failures.

Golden master testing tries to overcome these issues through several measures.
Rather than having a manually-defined test oracle,
it is derived from the SUT~\cite{barr15}.
This also protects against unexpected changes because the entire GUI can be checked at once.
Industrial tools are available both open source (e.g. Depicted~\cite{slatkin16}) and commercial (e.g. Percy~\cite{percy19}).
We find that Applitools~\cite{applitools19a} is currently the most sophisticated industrial implementation.
CV algorithms attempt to only report perceptible differences,
an advantage over conventional pixel-by-pixel approaches.

Similar approaches exist in academia.
FieryEye~\cite{mahajan16}, for example, expects three inputs:
the web page to test, an \enquote{appearance oracle}
(i.e. a golden master screenshot of a previously correct version) and
an optional list of regions to ignore.
The tool combines perceptual image differencing~(PID) from WebSee~\cite{mahajan15} with
a probabilistic model to link visual differences to possible root causes.
Browserbite~\cite{saar16} and X-Pert~\cite{roychoudhary13} are comparable,
although these approaches focus on cross-browser testing.
That is, instead of a golden master from a previous version,
the page is loaded in two different browsers and then checked for differences.

But as with most tools based on image processing,
they tend to report false positive when minor or unimportant changes occur~\cite{alegroth18}.
The user can ignore certain areas on the screen, but then all information within this area gets lost.
With the AGS filter mechanism,
the user is able to precisely specify what to ignore and what not.
We can even filter out the entire text of a GUI to test the effects of internationalization to the visual appearance,
for which image processing is typically too fragile~\cite{alameer16}.
Furthermore, most pixel-based approaches have problems grouping together same or similar changes.
This confronts the user with more maintenance effort than necessary.

Tools like ReDeCheck~\cite{walsh17} aren't pixel-based,
but are generally only suitable for a subclass
of visual bugs (e.g. cross-browser layout issues) due to the oracle problem.
Such an approach cannot decide if, for instance, a color or font is correct.
Without an extensive specification, only a human can do this,
which is why it is important to provide useful diagnostic information and
facilities to speed up the manual verification process.
Basically, the AGS can be used to statically detect layout failures too,
and report these to tools such as Viser~\cite{althomali19} for further verification,
which we aim to study as part of our future work.

Another notable approach is GCat~\cite{moran18}.
GCat takes two commits from a version control system~(VCS) like Git.
For both commits, GCat automatically explores the GUI and extracts screenshots as well as GUI metadata.
The resulting data is processed to match screens between the two GUI versions and to filter duplicates.
Similar to the AGS, GCat constructs a tree-based representation of the GUI.
But after corresponding leaf-element pairs are identified, PID is used for change detection.
Each change is further analyzed and categorized, and a natural language summary is created.
The tool is implemented for Android, but only supports comparisons of screens captured on the same device.
As the evaluation shows, reasons for FPs are, e.g., ambiguities related to font changes.
The AGS filter mechanism allows the user to easily ignore such changes
by filtering the \texttt{font} attribute.

Similar to GCat and our approach is Guide~\cite{xie09},
a platform- and language-independent tool for differencing GUIs.
It obtains information about the GUI using the OS accessibility layers.
GUI states are represented as trees,
where nodes describe composite GUI elements (e.g. layout containers) and
leaves are primitive GUI elements (e.g. buttons).
Each GUI element is abstracted by a set of properties including their values.
The original mapping algorithm computed a match score
for each GUI element pair,
similar to $\maintainedByIII{}{}$.
Grechanik et al.~\cite{grechanik18} extended Guide and
used various tree-edit distance algorithms to
improve the GUI mapping precision.
Although Guide offers a GUI to review results,
it is not clear how useful the resulting diagnostic information is.
Furthermore, the effectiveness of Guide was evaluated
using mostly artificial GUIs---where running times reach up to \SI{3}{\hour}---and
not a large set of real-world websites.

\section{Conclusion \& Future Work}
\label{sec:conclusion-future-work}

We have presented a platform-independent abstract GUI state~(AGS),
for which we have defined structural relations to identify relevant GUI changes.
The AGS framework can be used to perform cross-browser and \=/platform visual testing and
allows to ignore changes unimportant for the user.
We also explored various strategies on top of the AGS
for identifying deleted, created and maintained elements
to provide useful diagnostic information.
Experiments showed that our implementation can
effectively and efficiently be used to visually test web GUIs,
and that we are able to compete with a sophisticated industrial tool.

In future work, we want to further improve our current implementation and
strive for broader platform support, especially on mobile,
to compare states across different GUI technologies.
We also aim to investigate how the AGS can be used
for other GUI testing-related research questions,
such as the GUI element identification problem for test script repair.
Furthermore, we want to explore how we can combine autonomous GUI exploration with our approach.

\begin{acks}
As part of the joint research project \enquote{Surili},
this work is supported by a grant (no. 01IS17092A) from
the German Federal Ministry of Education and Research.
\end{acks}

\printbibliography

\appendix

\section{AGS}

Further details on the AGS itself.

\subsection{Equality vs. Inequality}

We give a self-contained canonical definition of AGS equality.
The purpose is to ensure that our notion of inequality
is sound and complete w.r.t.~the canonical AGS equality relation.

\begin{definition}[Equality]

\bda{c}
\rlabel{Es-Eq1} \ \equality{[]}{[]}
  \qquad \
  \rlabel{Es-Eq2} \
  \myirule{ Es = [ E_1, \dots, E_n ]
          \\ Es' = [ E_1', \dots, E_n' ]
          \\ \equality{E_i}{E_i'} \ \mbox{for $i=1,\dots,n$}
          }
          {\equality{Es}{Es'}}

  \\
  \\
  \rlabel{As-Eq1} \ \equality{\{\}}{\{\}}
  \qquad \
  \rlabel{As-Eq2} \
  \myirule{ As = \{ A_1, \dots, A_n \}
          \\ As' = \{ A_1', \dots, A_n' \}
          \\ \equality{A_i}{A_i'} \ \mbox{for $i=1,\dots,n$}
          }
          {\equality{As}{As'}}  

  \eda
  \bda{c}
  \ba{cc}          
  \rlabel{E-Eq} \
  \myirule{ \equality{As}{As'}
           \\ \equality{Es}{Es'}
          }
          {\equality{(As, Es)}{(As', Es')}}

  &

  \rlabel{A-Eq} \
  \myirule{ K = K' \\ V = V' }
          {\equality{(K, V)}{(K', V')}}
  \ea
  \eda

\end{definition}
Rule \rlabel{Es-Eq2} checks for equality among two lists of elements
by checking for equality among the elements at the respective positions.
Rule \rlabel{Es-Eq1} represents the (base) case where
both lists are empty.
For attributes, we establish equality among the individual elements.
See rules \rlabel{As-Eq1} and \rlabel{As-Eq2}.
Equality among attributes holds if their key-value pairs are identical.
See rule \rlabel{A-Eq}.
In essence, two AGSes are equal if the order of elements remains intact
and the sets of attributes contain the same (identical) attributes.

Consider the following contrived example.
\bda{ll}
& \equality{[(\{ A_1, A_2 \}, [])]}
       {[(\{ A_2, A_1 \}, [])]}
\\
\rightarrow_{\rlabel{Es-Eq2}} &
  \equality{(\{ A_1, A_2 \}, [])}
       {(\{ A_2, A_1 \}, [])}
\\
\rightarrow_{\rlabel{E-Eq}} &
(1) \ \equality{[]}{[]} 
\\ &
(2) \ \equality{\{A_1, A_2 \}}
           {\{ A_2, A_1 \}}
\\
\rightarrow_{\rlabel{As-Eq2}} &
(3) \ \equality{A_1}{A_1} 
\\ &
(4) \ \equality{A_2}{A_2} 
\eda
Each resolution step is indicated via an arrow ($\rightarrow$),
where the arrow is labeled with the inference rules involved.
Subgoals are labeled via distinct numbers.
As we leave out the details of attributes (key-value pairs),
subgoals $\equality{A_i}{A_i}$ represent a base case.

Inequality implies that equality is not derivable and vice versa.
For this statement to actually hold we need a technical requirement.
We assume that AGSs are well-formed
which is generally the case.

\begin{definition}[Well-formed AGS]
  We say an AGS $G$ is \emph{well-formed} iff
  for each set $As$ of attributes in $G$
  we have that the keys in the set $As$ are distinct.
\end{definition}
For example, $[(\{ \attribute{k}{x}, \attribute{k}{y} \}, [])]$ is not well-formed
whereas $[(\{ \attribute{k}{x} \}, [( \{\attribute{k}{y} \}, [])])]$ is well-formed.
In general, the AGSes we deal with are always well-formed.

Assuming that our AGSes are well-formed,
we can achieve a canonical form for attributes
by sorting according to their keys.
See rule \rlabel{As-Eq2} where we can
reorder of the set of attributes.

\begin{proposition}
  Let $G$ and $G'$ be two well-formed AGSes.
  Then, we have that $\equality{G}{G'}$ is not derivable iff
  $\inequality{G}{G'}$ is derivable.
\end{proposition}
\begin{proof}
We provide a proof sketch.
  
Suppose $\inequality{G}{G'}$. We consider the various cases.
Consider rule \rlabel{Es-Ineq}. If $n \not= m$ then we can immediately
conclude that $\equality{Es}{Es'}$ is not derivable.
Otherwise,  for some $i$ we have that $\inequality{E_i}{E_i'}$
and via some inductive argument we conclude that $\equality{E_i}{E_i'}$ is not derivable.
Similar arguments apply to rule \rlabel{E-Ineq}.

Consider rule \rlabel{As-Ineq}. If we find some $j$ ($m \leq 1$) such that $\inequality{B_j}{B_j'}$
then due to rule \rlabel{A-Ineq} we can immediately conclude that $\equality{As}{As'}$
is not derivable. If there is no such $j$ then ${\mathit max}(p,q) \leq 1$.

Suppose $p \leq 1$. Then, there is some $(K_i, V_i)$ that is not matched by
any attribute in $As'$. Hence, we conclude that  $\equality{As}{As'}$ is not derivable.
The same reasoning applies for $q \leq 1$.
This concludes the proof for the direction from right to left.

Suppose $\equality{G}{G'}$ is not derivable.
Via similar reasoning as above we can show that $\inequality{G}{G'}$.
To establish rule \rlabel{As-Ineq} in case $\equality{As}{As'}$ is not derivable,
we require well-formedness of AGS.
Thus, we can assume that either failure is to some mismatched values, see rule \rlabel{A-Eq},
or we can collect the mismatched keys by the sets
$\{K_1,\dots,K_p\}$ and $\{K'_1,\dots,K'_q \}$.
\end{proof}

\subsection{Fairly Similar}

We establish some properties assuming
that the fairly similar relation is an equivalence relation.

We write $M_{G, G'}$ to denote the set of maintained elements
for $G$ and $G'$.

\begin{proposition}
  Let $G, G', G''$ be three AGSes.
  Let $(E,E') \in M_{G, G'}$ and $(E',E'') \in M_{G',G''}$.
  Then, $(E,E'') \in M_{G, G''}$.
\end{proposition}
\begin{proof}
  By assumption we have that
  $E \in \extract(G), E' \in \extract(G'), E'' \in \extract(G'')$
  where $\maintainedBy{E}{E'}$ and $\maintainedBy{E'}{E''}$.
  By transitivity we find that $\maintainedBy{E}{E''}$.
  Hence, $(E,E'') \in M_{G, G''}$. and we are done.
\end{proof}

The above guarantees that for expected $G$, if actual $G'$ becomes expected
for another actual $G''$, then the set of maintained elements remains stable.

\begin{proposition}
  Let $G, G', G''$ be three AGSes.
  Let $(E,E') \in M_{G, G'}$ and $(E,E'') \in M_{G,G''}$.
  Then, $(E,E'') \in M_{G, G''}$.
\end{proposition}
\begin{proof}
  By assumption we have that
  $E \in \extract(G), E' \in \extract(G'), E'' \in \extract(G'')$
  where $\maintainedBy{E}{E'}$ and $\maintainedBy{E}{E''}$.
  By symmetry and transitivity we find that $\maintainedBy{E'}{E''}$.
  Hence, $(E',E'') \in M_{G', G''}$. and we are done.
\end{proof}

\section{Implementation}

More details on our implementation.

\subsection{Mapping Web GUIs to AGS}

The mapping from a platform-specific GUI state to the AGS is a crucial part as it affects the precision of our approach.
In general, there is some information loss.
One could try to obtain as much data as possible and then map this to the AGS.
However, this would lead to too much noise and flaky tests.
For instance, certain CSS attributes are browser-specific;
so what is present in one browser may be missing in another.
This would lead to inequal attributes.
Although the user can always filter elements and attributes via $\filter$,
an adapter is supposed to provide sensible defaults that handle such platform-specifics.

To avoid this situation, we currently extract all available HTML attributes, but only a subset of CSS attributes;
namely only non-redundant attributes that influence the visual appearance of the GUI.
To further reduce the amount of data and avoid noise, we automatically filter out default attribute values.
In the case of web-based GUIs, there are, e.g., several CSS defaults, depending on the used HTML version and browser.
Only if an attribute value is not a default, it will be part of the resulting AGS.
In addition to the HTML and CSS attributes, we compute further attributes as mentioned in Section~\ref{subsec:syntax}.
This includes the path of an element, which corresponds to its absolute XPath.
We also include an element's $x$ and $y$ position as well as its height and width, relative to the current viewport.

\begin{figure}[h]
\begin{algorithmic}[1]
	\Function{ComputeAttributes}{$H$, $C$}
		\State $M \gets \emptyset$ \Comment{initialize resulting mapping}
		\State $D \gets$ current DOM from browser
		\For{each node $n \in D$}
			\State $p_n \gets$ absolute XPath for $n$
			\State $h_n \gets$ non-default HTML attributes $H$ for $n$
			\State $c_n \gets$ non-default CSS attributes $C$ for $n$
			\State $a_n \gets h_n \cup c_n$ \Comment{merge attribute sets}
			\State $M \gets M \cup (p_n, a_n)$ \Comment{put key-value pair}
		\EndFor
		\State \Return $M$
	\EndFunction
\end{algorithmic}
\caption{Pseudocode to compute attributes.}
\label{algo:js}
\end{figure}

\begin{figure}[h]
\begin{algorithmic}[1]
	\Function{ConstructAgs}{$M$}
		\State $G \gets \emptyset$ \Comment{initialize resulting AGS}
		\State $M \gets sort(M)$ \Comment{sort from root to leaves}
		\For{each key-value pair $m \in M$}
			\State $p_m \gets key(m)$ \Comment{get absolute XPath}
			\State $a_m \gets value(m)$ \Comment{get merged attribute sets}
			\State $Es_m \gets \emptyset$ \Comment{initialize child elements}
			\State $As_m \gets p_m \cup a_m$ \Comment{construct attributes}
			\State $E_m \gets (Es_m, As_m)$ \Comment{construct element}
			\State $P_m \gets$ parent element for $E_m$ via $p_m$ in $G$
			\If{$P_m \neq \epsilon$}
				\State $Ps_m \gets$ child elements for $P_m$
				\State $Ps_m \gets Ps_m \cup E_m$ \Comment{add $E_m$ to children}
			\EndIf
			\State $G \gets G \cup E_m$ \Comment{add $E_m$ to AGS}
		\EndFor
		\State \Return $G$
	\EndFunction
\end{algorithmic}
\caption{Pseudocode to construct the AGS.}
\label{algo:peer-converter}
\end{figure}

To trigger the adapter, the user creates a \emph{visual checkpoint} (see below) just like with Applitools.
A checkpoint is similar to an assertion in test code, but there is no need to specify expected and actual.
Instead, we automatically extract the current GUI state via the WebDriver API and then compare it to the persisted golden master.
Extracting the state data is done by executing JavaScript code inside the browser.
This code traverses the DOM and computes the aforementioned attributes for all nodes as shown in Figure~\ref{algo:js}.

The result is a mapping $M$ for each node $n$ from the absolute XPath $p_n$ to the node's attributes $a_n$.
The absolute XPath serves two purposes.
First, it acts as an ID since it is unique within a GUI state.
Second, until we have fully created the AGS, it is the only structure-preserving information we have.
Next, we use the resulting map $M$ to construct the AGS with the recheck API, which provides the basic AGS types ($G$, $E$, $A$, etc.).
This construction is described in Figure~\ref{algo:peer-converter}.
We first sort the mapping from the root to the leaves.
For each key-value pair $m$ in $M$, we construct an element $E_m$ including an empty list of child elements $Es_m$ and the computed set of attributes $As_m$.
We then look for the parent element $P_m$ and add $E_m$ to its children $Ps_m$.
Finally, we add $E_m$ to the GUI state $G$.

\subsection{Checkpointing}

In general, a test case $C$ on the GUI level can be seen as a sequence of actions $a$ that lead to certain states $s$:
\begin{equation*}
C := \langle s_0 \xrightarrow{a_1} s_1 \xrightarrow{a_2} \ldots \xrightarrow{a_n} s_n \rangle
\end{equation*}
In our case, the user executes the actions via the Selenium WebDriver API to stimulate the SUT.
Whenever the SUT has reached a GUI state the user wants to check, the recheck API can be used to create a \emph{visual checkpoint}.
We denote this by underlining the corresponding state $\underline{s_i}$.
A checkpoint triggers the adapter to extract the current GUI state and map it to the AGS.
When a checkpoint is reached for the first time, we use the resulting AGS as the golden master (expected).
Subsequent test runs then create to AGSes (actual) that are compared to this golden master.
That is, every time a test is executed, we compare the expected GUI state against the actual GUI state to detect changes.
Similar to Applitools, each checkpoint requires a so-called \enquote{step name}, which is used as an ID for the golden master.

Let's again pick up the demo login screen from Figure~\ref{fig:login-all}.
We want to perform two simple visual checks: one before and one after the login.
A corresponding test case could look like this:
\begin{equation*}
C := \langle \underline{s_0} \xrightarrow{a_1} s_1 \xrightarrow{a_2} s_2 \xrightarrow{a_3} \underline{s_3} \rangle
\end{equation*}
Where:
\begin{itemize}
	\item $s_0$: State with empty login screen.
	\item $a_1$: Insert username into corresponding text field.
	\item $s_1$: State with inserted username.
	\item $a_2$: Insert password into corresponding text field.
	\item $s_2$: State with inserted username and password.
	\item $a_3$: Click login button.
	\item $s_3$: State after login.
\end{itemize}
We assign the step name \str{before-login} to $s_0$ and \str{after-login} to $s_3$, and create a set $S$ of golden masters or GUI states),
respectively, which is associated with this particular test case:
\begin{equation*}
S := \{ G_{s_0}, G_{s_3} \}
\end{equation*}
Where:
\begin{align*}
G_{s_0} & \Leftrightarrow \str{before-login} \\
G_{s_3} & \Leftrightarrow \str{after-login}
\end{align*}

When the test case is run again, we try to assign the expected and actual version of the GUI state by using the provided step name.
Two situations can occur:
\begin{enumerate}
	\item We cannot find a golden master for the given step name,
			so we create one and then we fail the test case.
	\item The step name corresponds to a golden master,
			so we can proceed with the comparison of expected and actual.
\end{enumerate}

In any case, the user is responsible to ensure that the SUT has finished loading.
Only then we know that both expected and actual GUI state are in a steady state and, therefore, comparable.

\end{document}